\documentclass[prl,aps,twocolumn,superscriptaddress,dvipsnames,amssymb,9pt,floatfix]{revtex4}
\usepackage[normalem]{ulem}
\usepackage{graphicx}
\usepackage{amsmath,amssymb,amsthm,amsfonts,mathptmx}
\usepackage{xcolor}
\usepackage{algorithm}
\floatname{algorithm}{Protocol}
\usepackage[labelsep=period]{caption}
\usepackage{ftnxtra}

\usepackage[T1]{fontenc}
\usepackage{mathptmx}

\newcommand{\add}[1]{{\color{blue}#1}}

\theoremstyle{plain}

\newtheorem{theorem}{Theorem}
\newtheorem{corollary}{Corollary}
\newtheorem{lemma}{Lemma}
\theoremstyle{definition}
\newtheorem{definition}{Definition}

\newcommand{\ket}[1]{|#1\rangle}

\begin{document}
\title{Quantum Correlations Can Speed Up All Classical Computation}
\author{Carlos A Perez-Delgado}
\email{c.perez@kent.ac.uk}
\affiliation{School of Computing, University of Kent, Canterbury, Kent CT2 7NF United Kingdom.}
\author{Sai Vinjanampathy}
\email{sai@phy.iitb.ac.in}
\affiliation{Department of Physics, Indian Institute of Technology Bombay, Powai, Mumbai-400076, India.}
\affiliation{Centre for Quantum Technologies, National University of Singapore, 3 Science Drive 2, Singapore 117543.}

\begin{abstract}
Quantum algorithms that can speed up certain tasks, such as factorisation and unstructured search, have driven a decades-long development of quantum computers and quantum technologies. Yet, outside specialized applications, quantum computers are believed to offer no advantage over classical computers. Here, we present a method which exploits quantum effects to speed up all possible classical computations. This method---which we call Coherent Parallelization (CP)---exploits quantum correlations generated by higher-order Hamiltonians to speed up any possible classical computation by a factor that depends on the classical algorithm. This factor is quadratic in the size of the input for a large set of interesting problems, leading to a strong commercial application in the emergent area of quantum technologies. We present important theoretical consequences of CP for both quantum physics and the theory of algorithmic complexity and computability, and discuss how CP can be implemented using real-world systems with natural or engineered Hamiltonians.
\end{abstract}

\maketitle


The advent of quantum computation revealed an implicit assumption in all models of computation---namely, that the computer, its state and its evolution, were all classical. More importantly, it showed that moving beyond this assumption would allow for the efficient solution to problems previously thought intractable\cite{shor99}. The way quantum computers are generally shown to obtain an advantage over their classical counterparts is through the use of quantum algorithms. Quantum algorithms use quantum operators (gates) that are inaccessible to their classical counterparts and that can be composed together to create (quantum) computational circuits that are more efficient at solving particular problems.

Cost analysis, both for quantum and classical computation, is done entirely at the level of abstraction of circuits and gates. Circuits---both classical and quantum---are all composed of one- and two-(qu)bit gates, all of which are assumed to take an equal standard time to operate. The quantum advantage, in this case, rests solely on the fact that quantum circuits require fewer quantum gates to solve certain problems compared to their classical analogues. By going beyond the abstraction of quantum computation as circuits consisting of one- and two-body gates, a different quantum advantage is accessible. This involves going one level of abstraction \emph{below} that of  operators (gates) and studying the Hamiltonians by which these are implemented.

By exploiting many-body interactions in the Hamiltonian \cite{beltran,boixo,roy} and by the use of non-linear Hamiltonians, a quantum advantage has been shown in metrology  \cite{dorner2009optimal,giovannetti2006quantum,Kok2010,giovannetti2011advances,degen2017quantum}, imaging  \cite{perez2012,brida2010experimental} and energy storage  \cite{Binder2015,Campaioli2017,ferraro2018high,le2018spin}. In the case of metrology, this method of obtaining quantum advantage has now been experimentally demonstrated \cite{napolitano}. In this paper, we prove that a minimal model of non-linear Hamiltonians can be used to provide a quantum advantage in the implementation of Toffoli logic gates. Since Toffoli gates are universal gates for reversible classical computation, our method can be used to speed up all classical computation.

Our method does not change the \emph{total gate complexity} of a problem. Rather, it acts much in the same way as (classical) parallelism. Parallelization does not reduce the total gates needed to solve a problem, but it may reduce the total \emph{time} needed to do so by reducing the \emph{amortized} time needed to implement each gate. This speed-up can be drastic in some cases: the difference between gate complexity and \emph{depth} complexity can be as high as a single polynomial order. Our method acts similarly, but it exploits quantum coherences to speed up classical computation. Hence, we name our method \emph{Coherent Parallelization} (CP).  Our method can reduce the time required to solve a problem by up to two polynomial orders compared to a standard classical implementation (\emph{e.g.} a reduction from $O(n^3)$ to $O(n)$ time). From a theoretical standpoint this is the first time a quantum computer has been shown to have a computational time-cost advantage over classical computers \emph{for any possible computation}.

\section*{Results}

The main contribution of this paper is a method that exploits the quantum correlations introduced by higher-order Hamiltonians in order to speed-up classical computation. By extending a well-understood technique used in metrology \cite{beltran,boixo,roy,napolitano} and energy storage \cite{Binder2015,Campaioli2017}, we show that quantum correlations can impact all classical computation.

In any physical computing device, the state of the system is evolved from an input state to an output state under the influence of a \emph{computational Hamiltonian} $H$. Further, for this device to be physical, the semi-norm of the Hamiltonian $p(H)$ must be bounded by a constant, where we define

\begin{align}
p(H) = \vert\langle H - h_{\mbox{min}}I \rangle\vert = h_{\mbox{max}} - h_{\mbox{min}}.
\end{align}

Here $h_{\mbox{max}}$, $h_{\mbox{min}}$ are the maximum and minimum eigenvalues of $H$ respectively and $I$ is the identity operator. This semi-norm quantifies the well understood tradeoff between the energy-per-second cost of a computation and the time-cost of a computation. 

All quantum evolution is bounded by the \add{q}uantum \add{s}peed \add{l}imit (QSL) \cite{taddei2013quantum,deffner2017quantum,campaioli2018tightening}. This limit  states that the time $\tau$ to transform any quantum state $\rho_0$ to another state $\rho_f$ under a physical map is no smaller than $\tau_{QSL}$, given by
\begin{align}\label{eq:qsl}
\tau\geq \tau_{qsl}:=\mathcal{L}(\rho_0,\rho_f)\max\left(\frac{1}{E},\frac{1}{\Delta E}\right).
\end{align}

 Here $\mathcal{L}(\rho_0,\rho_f)$ is the Bures angle between the two states, $E=\tau^{-1}\int_{0}^{\tau}dt\langle H(t)-h_{min}\rangle$ is related to the average energy  \cite{margolus1998maximum} and $\Delta E=\tau^{-1}\int_{0}^{\tau}dt\Delta H(t)$ is related to the average standard deviation of energy  \cite{mandelstam1991uncertainty}. Here $h_{min}$ is the instantaneous ground state of the time-dependent Hamiltonian and $\Delta H(t)$ is the instantaneous standard deviation of the energy. This motivates our use of the semi-norm in our resource accounting. In short, the QSL gives a direct trade-off between the semi-norm $p(H)$ and the time it requires to perform a particular evolution using that Hamiltonian.

To exemplify this tradeoff explicitly, consider the logical NOT gate. This can be implemented quantum mechanically by the application of the unitary operator $U = \sigma_x$. The unitary operator in turn can be implemented by setting the Hamiltonian of the qubit system to $H = \sigma_x$ for a time $t = \pi/2$. If instead we use the Hamiltonian $H = 2 \sigma_x$, then we reduce the time needed to complete the evolution of a $NOT$ gate in half. In general the time $t$ scales in inverse proportion to the semi-norm $p(H)$.

In the context of computation we can see that by increasing the value of $p(H)$ where $H$ is the Hamiltonian driving the evolution of a quantum register during the implementation of a quantum gate, any computation can be sped up---almost \footnote{As Lloyd points out, this speed-up is not \emph{quite} arbitrary \cite{lloyd}. If one increases the instantaneous energy of the system too high, one runs the risk of creating a black hole where one's computer used to exist. } arbitrarily so---at the trade-off cost of increasing instantaneous energy of the system, or energy-per-second.

Hence, in order for time complexity of algorithms to remain meaningful within our model we must set a limit $p\left(H(t)\right) \leq \tau$, where $\tau(n)$ is a constant that can scale at most linearly in the size of the input $n$. Within this semi-norm constraint, we will discuss how we can make use of quantum correlations generated by non-linear Hamiltonians to speed up computation. We first discuss an example of CP applied to the NOT gate, followed by the formal theorem.

The traditional way of implementing two instances of the NOT gate in parallel is to apply the unitary operator $\sigma_x = e^{-i  \pi/2 \sigma_x }$ to each individual qubit. Collectively the system's Hamiltonian is set to $H_{\|} = \sigma_x \otimes I + I \otimes \sigma_x$, and evolved for time $t = \pi/2$.

\begin{figure}[t]
\begin{center}
\includegraphics[width=0.44\textwidth]{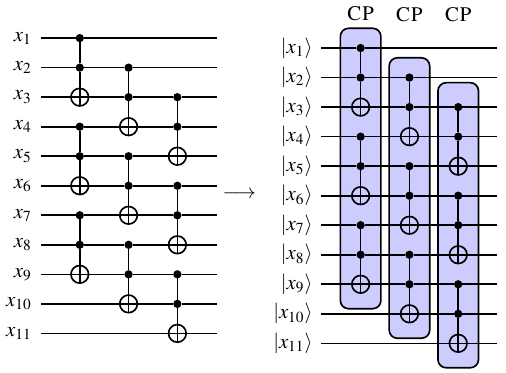}
\end{center}
\caption[Coherent Parallelization]{\emph{Coherent Parallelization:} 
The left hand side is a classical circuit consisting of nine Toffoli gates, arranged into three layers with three gates each. It requires nine time-steps to perform this circuit sequentially. A classical parallel implementation can perform the same circuit in three time steps---each gate having an amortized cost of $1/3$ time-steps. A coherent parallel implementation, right hand side, can further reduce the amortized cost of each gate down to $1/9$ time-steps, and hence perform the same circuit in a single time-step.
}
\label{fig-nots}
\end{figure}

Another way to implement this joint operation is in a coherent fashion. Instead of the Hamiltonian used above we use $H_{\#} = \sigma_x \otimes \sigma_x$. Note that $\sigma_x \otimes \sigma_x$ $ =$ $e^{-i  \pi/2 (\sigma_x \otimes I + I \otimes \sigma_x)}$ $=$ $e^{-i  \pi/2 (\sigma_x \otimes \sigma_x)}$. So, in both cases the system needs to be evolved under the appropriate Hamiltonian for time $t = \pi/2$. However $p(\cdot)$ is a resource for state transformations (and hence computation), since the quantum speed limit $t_{QSL}$ bound depends on this quantity. Hence to fairly compare the parallel and the coherent strategy it is necessary  to fix the resources, namely $p(H)$. Since $p(H_{\|}) = 2$ whereas $p(H_{\#}) = 1$, we can \emph{scale} $H'_{\#} = 2 H_{\#}$, and stay within the same norm limit $\tau$ as the parallel implementation $H_{\|}$. Therefore, $H'_{\#}$ implements both NOT gates in half the time that $H_{\|}$ requires. This argument is at the heart of the quantum advantage in energy storage  \cite{Binder2015,Campaioli2017}.

In the following two results  we generalize and extend this result to \emph{all} operators that are both unitary and Hermitian (see Supplementary Material for proofs).

\begin{lemma}
Let $H$ be a Hermitian, unitary operator. Then, for any positive integer $m$:
\begin{align}
	p\left( H_{\|}(m) \right) &= m p\left(H\right) \\
	p\left( H_{\#}(m) \right) &= p\left(H\right)
\end{align}
\end{lemma}

Here $H_{\|}(m)$ ($H_{\#}(m)$) refers to the coherent (standard) parallel implementation of $m$ copies of the unitary, Hermitian operator $H$. The following theorem follows directly from the previous lemma.

\begin{theorem}[Coherent parallelization]
	Let $H$ be any Hermitian unitary gate acting on a $d$-dimensional system or qudit. Implementing $m$ gates in parallel using a standard parallel computation implementation $H_{\|}$ is $m$ times slower than using a coherent parallelization approach $H_{\#}$.
\end{theorem}

The Toffoli gate is both a unitary and Hermitian operator. Moreover, it is universal for (reversible) classical computation. Hence, any classical reversible computation on $n$ bits can be implemented as a circuit consisting of $d$ layers, each consisting of $1 \leq m \leq n$ Toffoli gates. At each layer (time-step) one may choose to implement the Toffoli gates sequentially, in parallel, or coherently together using the method we described above. Using this latter method allows us to reduce the amortized time cost of implementing each individual Toffoli gate to  $1/m$ times the amortized (classical) parallel amortized cost of implementation, or to $1/m^2$ the standard sequential time cost of implementation.

This advantage is maximized in the case of highly parallelizable reversible circuits (where $m \approx n$ at every depth). In such cases the above method has the effect of reducing the time required to run the algorithm by two polynomial orders. For example, an $O(n^3)$  algorithm can be performed with $O(n)$ resources, an $O(n^2 \log n)$ algorithm with $O(\log n)$ resources \emph{etc.})
We arrive at the following result:

\begin{corollary}
	Let $\{C_n\}_n$ be  a uniform family of reversible circuits. The same  computation can be performed using coherent parallelization in time $T_{\#}(n) = O \left(\mathcal{C}(n) / \Delta^2 \right)$ where $\Delta  = \mathcal{C}/\mathcal{D}$.
\end{corollary}	

Here $\mathcal{S}$ and $\mathcal{D}$ refer to the circuit-cost complexity and circuit-depth complexity respectively of the algorithm in question.

Any \emph{irreversible} classical circuit over the gate set $\{$NAND$\}$ can be transformed into a reversible one over the gate set $\{$Toffoli$\}$ using one of many techniques  \cite{bennett1,bennett2,amy2017}. Hence, the above technique can be applied to \emph{any} classical algorithm. Bennett's method to convert irreversible computation to reversible computation, which consists of replacing all NAND gates with Toffoli ones (and introducing ancillary bits) neither increases the computational time complexity, nor does it change the computational depth complexity \footnote{While it does increase the space complexity, this is irrelevant to our analysis here.}. 
With this, we state our final theorem:

\begin{theorem}[Coherent parallelization of classical circuits]
	Let $\{C_n\}_n$ be  a uniform family of classical circuits over the universal gate set $\{$NAND$\}$. The same  computation can be performed using coherent parallelization in time $T_{\#}(n) = O \left(\mathcal{S}(n) / \Delta^2 \right)$ where $\Delta  = \mathcal{S}(n)/\mathcal{D}(n)$.
\end{theorem}	
Here $n$ is the size of the input, $\mathcal{S}(n)$ is the circuit gate cost (number of gates) of the classical circuit, and $\mathcal{D}(n)$ is the depth of the circuit. In short, coherent parallelization can be used to reduce the time cost of any classical algorithm to below its depth complexity.

\section*{Discussion}

Coherent parallelization (CP) exploits the same type of correlations in the Hamiltonian as do Heisenberg-limited metrology and quantum enhanced charging. As such, it is an intrinsically quantum effect with no classical analogue. 
CP is quite different from the advantage that well-known quantum algorithms have over classical. Algorithms like Shor's\cite{shor99} or Grover's\cite{grover96} display an advantage over classical counterparts when solving \emph{particular} problems. By comparison, CP can be used to accelerate \emph{any possible} classical computation. This speed-up is such that it can reduce the cost of solving a problem by up to two polynomial orders. This is clearly less than the advantage that Shor's algorithm has over classical factorization algorithms. However, coherent parallelization is a method that can be applied much more generally. It is much more general than even Grover's search algorithm\cite{grover96}.

Grover's algorithm can be used to solve any problem  within the class NP, however, it only gives an advantage for problems that do not (currently) have an algorithm that solves the problem more efficiently than brute-force search. Once a problem has an algorithm that solves it at least quadratically more efficiently than brute-force search, Grover's algorithm ceases to provide any advantage.
 On the other hand, CP is itself not an algorithm. Rather, it is a method that can accelerate any existing classical algorithm using quantum coherence. Hence, CP can provide a quantum advantage on \emph{any} possible computational problem, no matter how efficient current classical algorithms are.

A direct, important, consequence of our result then relates to provable quantum advantage. Previously, outside of oracle/black-box scenarios, the only provable computational advantage of quantum computation devices over classical was for a very narrow set of problems\cite{bravyi18}. Because CP can improve the run-time of \emph{any} classical algorithm, we now have provable quantum advantage for \emph{all} computational problems.

The proportion of this advantage grows the more parallelizable the problem is. The maximum advantage is achieved for problems that can be efficiently solved using parallel computation using low  (logarithmic) depth circuits---\emph{i.e.} problems in the class NC. Any such problem can be solved with coherent parallelization in logarithmic time. Hence  coherent parallelization  is particularly well-suited for speeding up ubiquitous mathematical tasks such as matrix multiplication and speed up physically important tasks such as \emph{Monte Carlo} simulations, genetic algorithms and many particle-physics simulations. Other computations that are particularly well-suited for coherent parallelization that are worth mentioning due to their real-world applications include machine-learning tasks like hyperparameter grid search and   cryptographic tasks such as   proof-of-work in crypto-currencies and blockchain technologies.

Let us consider sorting as a concrete example of the advantage obtainable using CP. Sorting $n$ integers can be done sequentially in $O(n \log n)$ time. This time is optimal for general integers. A fully classical parallel approach can do the same sorting in $O(\log n)$ time optimally. Using the same parallel algorithm, but using CP to speed up the implementation of gates allows us to sort these same $n$ integers in  $O\left(\frac{1}{n}\log n \right) \in O(1)$ time.

Another consideration is  error correction. A full systematic analysis of error correction is beyond the scope of this paper. However, it is worth noting that coherent parallelization, unlike (traditional) quantum computation,  implements solely classical gates, with classical inputs and classical outputs. In a traditional quantum computer, the quantum state of the computation, with all its coherences, must be maintained throughout the computation, from beginning to end. There is no such need in our scheme. Here, after each timestep, the state of the computer can be asserted to be entirely classical. Hence, only bit-flip errors need to be corrected.

Next, we comment on the accounting of resources associated with implementing CP. From a computational complexity perspective, our method acts similarly to classical parallelization. Both methods reduce the total time-cost of solving a problem, without reducing the computational complexity, by reducing the \emph{amortized} cost of implementing each individual gate. From a physical perspective, our use of higher-order Hamiltonians, and our resource-cost analysis is completely in-line with the use of higher-order Hamiltonians in other physical settings (besides computation) such as metrology and battery charging. In the case of metrology, super-Heisenberg metrology is a speed-up over other forms of quantum metrology achievable through the use of higher-order Hamiltonians.

Super-Heisenberg metrology has been theoretically shown to be possible only using higher-order Hamiltonians \cite{beltran,boixo,roy}.
There are various proposals for implementing these speedups using widely different experimental setups, such as scattering in Bose condensates \cite{boixo08}, Duffing nonlinearity in nano-mechanical resonators \cite{woolley}, two-pass effective non-linearity with an atomic ensemble \cite{chase}, Kerr-like nonlinearities \cite{rivas}, and nonlinear quantum atom-light interfaces \cite{napolitano10}. Finally, and most importantly, this speed-up has now been  experimentally demonstrated  \cite{napolitano}. In summary, it has been theoretically established and experimentally verified that it is indeed possible to speed-up a quantum process, \emph{without changing its circuit complexity,} by using higher-order Hamiltonians. CP extends this theoretically established and experimentally verified effect to speed up computation.

Coherent parallelization should be understood to be  be a result at the intersection of fundamental physics and theoretical computer science. It is a statement about the fundamental cost of performing a computation. 
Every quantum evolution including that of a computation is bounded by the quantum speed limit (see Eq.\ \ref{eq:qsl} and the discussion surrounding it). In short, the QSL tells us that (for purposes of quantum evolution) time \emph{is} energy. Reducing the semi-norm $p(H)$ of the Hamiltonian $H$ needed to implement a quantum computation is the same as reducing its time cost. And, this is precisely what CP allows us to do.

That said, the potential impact of this result goes well beyond theory. While an engineering analysis of CP is beyond the scope of this paper, it is likely that this result will lead to more efficient implementations of computation in practice. There are various potential ways to outright implement CP. Many systems with the desired collective quantum behavior have been studied beyond the ones already mentioned above---both natural  \cite{gross1982superradiance} and  engineered  \cite{Roy2017}. 
It has been shown that using many-body Hamiltonians that are natural to the system  being used to implement multi-qubit gates can yield advantages over implementations  that rely on naive gate decompositions  onto gates from a \emph{`standard'} universal gate-set\cite{Shende,zahedinejad2015high}.
Most importantly, recently Ferrero \emph{et.~al.} \cite{ferraro2018high} described a proposal for implementing what can be properly  seen as a CP implementation of a $NOT$ gate. A similar analysis, but for the Toffoli gate instead, would allow for a universal CP speed-up of all classical computation. This highly suggests that CP is not just of theoretical interest, but of great potential practical interest as well.

In closing, from a practical perspective CP would allow for properly designed quantum computers to speed up \emph{all possible} classical computation (rather than a small subset). This would drastically increase the interest in and impact of quantum computation. From a theoretical perspective, CP provides, for the first time, a clear tradeoff between the time required to perform a computation and quantum correlations. It opens a new approach of studying computation complexity that goes beyond circuit complexity to also consider quantum correlation complexity.

\section*{Acknowledgements}
The authors would like to thank Rosario Fazio, Felix Binder, Yingkai Ouyang for discussions and comments on early versions of this manuscript.

SV acknowledges support from an IITB-IRCC grant number 16IRCCSG019 and by the National Research Foundation, Prime Minister's Office, Singapore under its Competitive Research Programme (CRP Award No. NRF-CRP14-2014- 02).



\clearpage

\newpage

\section*{Supplementary Material}

\setcounter{theorem}{0}
\setcounter{lemma}{0}
\setcounter{corollary}{0}

We begin with a formal definition of the function $p(\cdot)$ that we have used in the main body of the paper.

\begin{definition}\label{def:p} 
Let $H$ be any Hermitian operator. Then
\begin{equation}
p(H) = |\langle H - h_{\mbox{min}}I \rangle| = h_{\mbox{max}} - h_{\mbox{min}},
\end{equation}
where $h_{\mbox{max}}$, $h_{\mbox{min}}$ are the maximum and minimum eigenvalues of $H$ respectively. 
\end{definition}
Next is a discussion of the model of computation we are presenting for the first time in this paper.

\emph{Model of Computation.---} 
We start this section with a formal definition of our computational model:

\begin{definition}[Computing Machine]\label{def:comp-model} 
	A Computing Machine (CM) consists of a closed physical system with three subsystems $B, C, S$: the \emph{battery, control,} and \emph{input/output} systems respectively.
\begin{description}
\item[Battery] consists of an unbounded countable number of two-dimensional subsystems each with Hamiltonian $H_B = \sigma_z$.
\item[Input/output] consists of a countably infinite dimensional system with Hamiltonian $H_0$ that can be arbitrarily chosen. All but a  finite subsystem $S$ of dimension $N = 2^n$ is set to the ground state of $H_0$ at the beginning of this computation. The state $\rho_0$ ($\rho_f$) of the subystem $S$ at the start (end) of the computation is called the \emph{input (output).}
\item[Control] exchanges energy with the battery subsystem to power the application of a Hamiltonian $H(t)$ for a time $T$ to the input/output system during computation. This can be done with standard energy conserving unitaries. This Hamiltonian is such that
\begin{equation}
	p\left(H(t)\right) \leq \tau(n), \quad \forall t,
\end{equation}
where $n$ is the size of the input and $\tau(n)$ is a constant that at most scales linearly in $n$.
\end{description}
\end{definition}

The purpose of our model of computation is to act as the most general abstraction of natural process that can perform computation, without ignoring any of the necessary physical properties of such a process.

The purpose of the battery subsystem is to account for the energy required to perform the computation. In order to be able to compare meaningfully different computations, a standard battery Hamiltonian is chosen for every possible computer and computation performed.

The input/output subsystem is how the computer communicates with the external world, and meaningfully performs computation. The subsystem is initialized to the input state before computation. At the end of the computation the subsystem should then hold the output state. An arbitrary Hamiltonian for this system is allowed in order to be able to model---and quantify the energetic resources in---computation on different information carriers. These information carriers can be anything from ions in a trap or potential well, to nuclei, to anyons, depending on the actual implementation of the quantum computer and they all different real-world Hamiltonians.

While we leave the possibility open in our model to any possible input/output system, we will be particularly interested in (and restrict further discussion to) systems with a homogenous repeating structure, \emph{e.g.} $n$ spin$-1/2$ particles each with Hamiltonian $\sigma_z$ and pairwise Ising interaction.

Finally, the sole purpose of the control subsystem is to provide a \emph{locus} for the computation itself. It mediates between the battery and the input/output system, and performs the computation itself by drawing power from the former, and applying an external Hamiltonian to the latter. This Hamiltonian $H(t)$ is time dependent, and arbitrarily chosen based on the computation to be performed. As mentioned in the main text, in order to maintain the meaningfulness of algorithmic time complexity within our model we must bound $p\left(H(t)\right)$ from above. Our bound $\tau(n)$ is dependent on $n$ to allow for parallel computation (performing multiple gates on different qubits at the same time). If we further limit $\tau$ to be a constant independent of $n$ we can define a \emph{sequential} computing machine. In this paper we focus on the more general (parallel) model.

There are many ways (computational models) to describe classical computations. Here we use the well understood standard circuit model. We understand a uniform family of reversible circuits $\{C_n\}_n$ to consist of  circuits $C_n$ consisting of only Toffoli gates, each acting on $n$ bits of input. Let $\mathcal{D}(C_n)$ be the circuit depth of $C_n$. For every $0 \leq i \leq N = 2^n$ let $C_n(i)$ be the result of running the circuit $C_n$ on the binary representation of $i$ as input. 

Also when discussing a particular algorithm described as a family of circuits, we will use $\mathcal{S}$, and $\mathcal{D}$ to refer to its circuit-cost complexity and circuit-depth complexity respectively.

\emph{Coherent Parallelization.---} We now focus on our method for increasing the efficiency/speed of arbitrary classical computations. 

Consider a quantum system with $m$ identical sub-systems (qubits, qudits or the tensor product thereof). 
Let $H[i]$, $1 \leq i \leq m$ for any Hermitian and/or unitary operator $H$ to mean $H$ applied to the $i'th$ subsystem. Formally:
\begin{equation}
H[i] = \left( \bigotimes_{i - 1} I \right) \otimes H \otimes \left( \bigotimes_{n - i - 1} I \right).
\end{equation}

For any  Hermitian, unitary operator $H$  we define
\begin{align}
H_{\|}(m) &= \sum_{i = 1}^m H[i]\\
H_{\#}(m) &= \bigotimes_{m} H
\end{align}

We then have the following result.

\begin{lemma}\label{lemma:coh-par}
Let $H$ be a Hermitian, unitary operator. Then, for any positive integer $m$:
\begin{align}
	p\left( H_{\|}(m) \right) &= m p\left(H\right) \\
	p\left( H_{\#}(m) \right) &= p\left(H\right)
\end{align}
\end{lemma}
\begin{proof}
Since $H$ is both Hermitian and unitary, its only possible eigenvalues are $\pm 1$. So either $p\left(H\right) = 2$, or $p\left(H\right) = 0 $. If $p\left(H\right) = 0$, then the lemma follows trivially. Therefore, lets assume $p\left(H\right) = 2$. Let $\ket{+}$ ($\ket{-}$) be $+1$ ($-1$) valued eigenket respectively of $H$. Then 
\begin{equation}
	\left(\sum_{i = 1}^m H[i] \right) \left( \bigotimes_m \ket{\pm} \right )= 
	 \pm m \left( \bigotimes_m \ket{\pm} \right )
\end{equation}
and
\begin{equation}
	\left(\bigotimes_{m} H \right)\left( \bigotimes_m \ket{\pm} \right )= 
	 \pm 1 \left( \bigotimes_m \ket{\pm} \right ).
\end{equation}
To complete the proof we must show that $\left( \bigotimes_n \ket{\pm} \right )$ are the maximum and minimum valued (respectively) eigenkets of both $\left(\bigotimes_{n} H \right)$ and $\left(\sum_{i = 1}^n H[i] \right)$. 

We show that the largest eigenvalue of $\bigotimes_{m} H$ is $1$. We proceed by contradiction. Assume there exists a vector $\ket{\omega}$ such that
\begin{equation}
	\left(\bigotimes_{n} H \right) \ket{\omega} = \omega \ket{\omega},
\end{equation}
where $\omega \in \mathbb{R}$ and $\omega > 1$, and that this is the largest valued eigenket of $\bigotimes_{m} H$.  Given that $\ket{\omega}$ is an eigenket of $\bigotimes_{n} H $ it must be that it may be written as 
\begin{equation}
	\ket{\omega} = \bigotimes_{i = 1}^n \ket{\omega_n},
\end{equation}
where each ket $\ket{\omega_n}$ is an eigenket of $H$. And further,
\begin{equation}
	\omega \ket{\omega} = \left(\bigotimes_{n} H \right) \ket{\omega}
	= \bigotimes_{i = 1}^n H \ket{\omega_n}
	= \bigotimes_{i = 1}^n 1 \ket{\omega_n},
\end{equation}
where in the last equality we used the facts that $\ket{\omega_n}$ is an eigenket of $H$, and that $H$ is both Hermitian and unitary. From here it follows that $1 < \omega = 1$, which is a contradiction. An identical argument can be used to show that $-1$ is the minimum eigenvalue of $H_{\#}$, and similar arguments can be used that the $\pm n$ are the maximum/minimum eigenvalues of $H_{\|}$.
\end{proof}

The following theorem follows directly from the previous lemma, and our computing machine definition.

\begin{theorem}[Coherent parallelization]\label{thm:coh-par}
	Let $H$ be any Hermitian unitary gate acting on a $d$-dimensional system or qudit. Implementing $m$ gates in parallel using a standard parallel computation implementation $H_{\|}$ is $m$ times slower than using a coherent parallelization approach $H_{\#}$.
\end{theorem}	
\begin{proof}
	Without loss of generality let the bound $\tau = 1$. Then, in order to use the standard parallelization method within the bound set, one must use a normalized version of the parallel Hamiltonian $H'_{\|}(m) = 1/m H_{\|}(m)$.
On the other hand to implement the $m$ gates using coherent parallelization one may use standard coherent parallelization Hamiltonian $H_{\#}(m)$ as defined above, since it is already normalized to $1$.
Then $p\left( H_{\|}(m)\right) = p\left( H_{\#}(m)\right) = 1$, as required. However,
\begin{equation}
	\bigotimes_m H = e^{ - i \pi/2 m H'_{\|}(m)} = e^{ - i \pi/2 H_{\#}(m)}.
\end{equation}
Which shows that using the Hamiltonian $H_{\#}(m)$ one can implement the desired gate $\bigotimes_m H$ a factor of $m$ times faster than using $H_{\|}(m)$.
\end{proof}

\begin{theorem}[Coherent parallelization of reversible circuits]
	Let $\{C_n\}_n$ be  a uniform family of reversible circuits, and let $A(n) = \{H(t), T\}$ be the implementation of said circuit as a computing machine algorithm and $T(n)$ is the time required to run $A$ on an input of size of $n$. The same  computation can be performed using coherent parallelization in time $ O \left( T(n)/ \Delta \right)$ where $\Delta  = \mathcal{S}/\mathcal{D}$.
\end{theorem}	
\begin{proof}
First we note that the average number of gates in $\{C_n\}_n$ at each depth $d$ is given by $\Delta  = \mathcal{S}/\mathcal{D}$. Hence, by Thm.\ \ref{thm:coh-par} the implementation of the gates of $\{C_n\}_n$ at depth $d$ can be sped up \emph{on average} by a factor of $\Delta$ using coherent parallelization over a standard implementation. Taking the behavior at the asymptotic limit as $n \rightarrow \infty$ gives us the desired result.
\end{proof}

Note that in the previous theorem we're comparing a coherent parallelization implementation to a standard computing machine implementation of a classical reversible circuit. However, this latter implementation is already parallel (all gates at any depth $d$ are taken to be implemented at once). Obviously, a parallel implementation has a speed factor advantage of $\Delta  = \mathcal{C}/\mathcal{D}$ over a sequential implementation.
We've hence proven the following corollary.

\begin{corollary}\label{cor:rev-circuit}
	Let $\{C_n\}_n$ be  a uniform family of reversible circuits. The same  computation can be performed using coherent parallelization in time $T_{\#}(n) = O \left(\mathcal{C}(n) / \Delta^2 \right)$ where $\Delta  = \mathcal{C}/\mathcal{D}$.
\end{corollary}	

We conclude with the following result.

\begin{theorem}[Coherent parallelization of classical circuits]
	Let $\{C_n\}_n$ be  a uniform family of classical circuits over the universal gate set $\{$NAND$\}$. The same  computation can be performed using coherent parallelization in time $T_{\#}(n) = O \left(\mathcal{S}(n) / \Delta^2 \right)$ where $\Delta  = \mathcal{S}(n)/\mathcal{D}(n)$.
\end{theorem}	
\begin{proof}
For this proof we first convert  $\{C_n\}_n$ to a reversible family of circuits that has both the same depth- and circuit-complexity, and then simply apply Corollary\ \ref{cor:rev-circuit}.
For the first step we use a result by Bennet  \cite{bennett1,bennett2} that states that any irreversible circuit family with space complexity $\mathcal{S}$, circuit depth complexity $\mathcal{D}$ and circuit complexity $\mathcal{C}$ can be perfectly simulated using a reversible circuit with 
space complexity $\mathcal{S} + \mathcal{C} $, circuit depth complexity $\mathcal{D}$ and circuit complexity $\mathcal{C}$.
\end{proof}

It is worth noting that there are many methods to convert an irreversible circuit into a reversible circuit all of which have a space/depth complexity tradeoff. For our purposes, Bennet's method is optimal as it allows us to reach the the theoretical optimal time performance for coherent parallelization. For many real-world applications it may be beneficial to consider newer irreversible-to-reversible transformation methods  \cite{amy2017}.

\end{document}